\documentclass[runningheads,a4paper]{llncs}

\usepackage{amssymb}
\usepackage{graphicx}

\usepackage{url}
\urldef{\mailsa}\path|dwagndwagn@gmail.com|    
\newcommand{\keywords}[1]{\par\addvspace\baselineskip
\noindent\keywordname\enspace\ignorespaces#1}

\newcommand{\defname} [1] {\noindent{\bf (#1) }\ }

\begin{document}

\mainmatter  %

\title{The Unified Segment Tree and its Application 
to the Rectangle Intersection Problem}

\titlerunning{The Unified Segment Tree}

\author{David P. Wagner \thanks{ \mailsa} }
\authorrunning{The Unified Segment Tree} %
\institute{ Hanyang University\\
Department of Electronics Engineering\\
Seoul, Korea\\
\url{www.hanyang.ac.kr/english}}

\toctitle{The Unified Segment Tree} 
\tocauthor{David P. Wagner}
\maketitle

\begin{abstract}
In this paper we introduce a variation on the multidimensional segment tree, 
formed by unifying different interpretations of the dimensionalities of the
data structure.  We give some new definitions to previously 
well-defined concepts that arise naturally in this variation, and 
we show some properties concerning the relationships between 
the nodes, and the regions those nodes represent.  We think these properties 
will enable the data to be utilized in new situations, beyond 
those previously studied.  As an example, we show that the data 
structure can be used to solve the Rectangle Intersection Problem 
in a more straightforward and natural way than had be done in the past.
\end{abstract}

\section{Introduction}

The Segment Tree is a classic data structure from computational
geometry which was introduced by Bentley in 1977 \cite{Bentley77}.
It is used to store a set of line segments, and it can be queried
at a point so as to efficiently return a list of all line segments
which contain the query point. 

The data structure has numerous applications.  For example, in its
early days it was used to efficiently list all pairs of intersecting 
rectangles from a list of rectangles in the plane \cite{BentleyW80},
to report the list of all rectilinear line segments in the plane which
intersect a query line segment \cite{VaishnaviW82}, and to report the
perimeter of a set of rectangles in the plane \cite{VitanyiW79}.
More recently the segment tree has become popular for use in pattern 
recognition and image processing \cite{RacherlaRO02}.

Vaishnavi described one of the first higher dimensional segment trees
in 1982 \cite{Vaishnavi82}.  Introducing his self-described 
``segment tree of of segment trees'', he attached an ``inner segment tree'',
representing one dimension, to every node of an ``outer segment tree'', 
representing the other dimension, and used this for the purpose 
of storing rectangles in the plane.  A point query would then return
a list of all rectangles containing the point.  The recursive nature of this
data structure meant that it could be generalized to arbitrary dimensions.
This has been the standard model for high dimensional segment trees ever since.

Here we describe a variation on the higher dimensional segment tree
which we feel represents a more useful way to store the data.  The new data 
structure is formed by merging the segment trees which could be formed 
by different choices for the dimensions of the inner and outer segment trees.
Our purpose in introducing this variation is not to show that it
is faster for a particular application, but to show that the data 
structure has some interesting properties, and therefore the data 
can be used in some new situations. 

In the following sections, we will define the data structure, and show 
a useful way to visualize it.  We introduce several new definitions as they 
apply to this variation of data structure.  We further show some 
interesting properties concerning the relationships between the nodes, 
and the regions those nodes represent.

Finally, we demonstrate that the data structure can be used to
solve the previously studied ``Rectangle Intersection Problem''.
Existing methods to solve this problem have either involved
using range trees, storing $d$-dimensional rectangles as $2d$-dimensional 
points, or sweep planes, processing a lower dimensional problem across
the sweep.  We think that our new data structure represents a more 
natural way to store this data, and the algorithms involved are more 
straightforward.

\section{Definitions and Properties}

Here we introduce a number of definitions, as they apply to Unified
Segment tree, including a definition of the data structure itself.

A description of the original segment tree is included in the 
Appendix, for the convenience of the reader.  Please note that many 
properties of segment trees are given therein, and these properties 
are referred to throughout this paper.

\subsection{The Unified Segment Tree}

We begin with a description of a two dimensional segment tree, 
given by Vaishnavi in 1982.  He describes this data structure, 
which stores rectangles in the plane, as a ``segment tree of segment trees'' 
\cite{Vaishnavi82}.  It begins with a single 
one-dimensional segment tree, which is called the ``outer segment tree'',
and which represents divisions of the plane along one of the two
axes.  Attached to every node of this outer segment tree, 
is an ``inner segment tree'', which represents further divisions 
along the other axis.

Here we note that the choice of axis for the outer segment tree
could have been either the x-axis or the y-axis.
Although this choice may be arbitrary, it has a great effect 
on the organization of the data structure.
Therefore, let us give different names to the different
data structures resulting from this choice.

\begin{definition}
\defname{$xy$-segment tree}
Define the $xy$-segment tree to be the two dimensional segment
tree whose outer segment tree divides the plane along the $x$-axis,
and whose inner segment trees further divide these regions along
the $y$-axis.
\end{definition}

\begin{definition}
\defname{$yx$-segment tree}
Define the $yx$-segment tree to be the two dimensional segment
tree whose outer segment tree divides the plane along the $y$-axis,
and whose inner segment trees further divide these regions along
the $x$-axis.
\end{definition}

The leaves created in a segment tree by the insertion of a segment
correspond to the sections of a canonical subdivision of the segment.
One leaf is created for each component of the subdivision.
Here we show that the same leaves are created regardless of whether 
an empty $xy$-segment tree or an empty $yx$-segment tree is used.

\begin{theorem}
The same rectangle inserted into an empty $xy$-segment tree and into 
an empty $yx$-segment tree will create an equivalent set of leaves in 
the two trees.
\end{theorem}

\begin{proof}
A rectangle inserted into an $xy$-segment tree is first
divided along the $x$-axis, and then each subregion is 
further subdivided along the $y$-axis.
In the $yx$-segment tree, the rectangle is first divided
along the $y$-axis, and then each subregion is further
subdivided along the $x$-axis.
The same subregions are created from the rectangle, regardless 
of the order in which the two axes are chosen, so therefore
an equivalent set of leaves is created in the two trees.
\end{proof}

If several rectangles are stored in a segment tree,
each is stored independent of the other rectangles.
So this leads to the following corollary.

\begin{corollary}
The same set rectangles inserted into an $xy$-segment tree and into 
a $yx$-segment tree will create an equivalent set of leaves in 
the two trees.
\end{corollary}

Now we can define the unified tree in two dimensions based on these 
two structures.

\begin{definition}
\defname{Unified Segment Tree -- 2 dimensions}
Define the unified segment tree storing a set of rectangles in the plane
to be the data structure created by the following procedure:
\begin{enumerate}
\item{Create both the $xy$-segment tree and the $yx$-segment tree containing
the set of rectangles.}
\item{Merge the root of every inner segment tree with the node of the
outer segment tree to which it is attached, so that they are considered
to be one node.}
\item{Merge any two nodes in the $xy$-segment tree and the $yx$-segment tree
which represent the same region of the plane, so that they are considered
to be one node.}
\item{Add all possible ancestors to any node which is missing any
of its available ancestors.
(Ancestors in this data structure are defined later.)}
\end{enumerate}
\end{definition}

We note several features of the new data structure which are not
normally associated with segment trees.

\begin{itemize}
\item{A node may have a two parents, one from the $xy$-segment tree and
one from the $yx$-segment tree}
\item{A node may have four children.  These could have been created when
the root of an inner segment tree was merged with a node of the outer
segment tree, or there could be two nodes each from the $xy$-segment tree
and from the $yx$-segment tree.}
\item{The new data structure is technically no longer a tree, as it
may contain cycles.}
\end{itemize}

\subsection{Parents, Children, Ancestors, Descendants}

In order to accommodate the features of the new data structure,
we must create some new definitions of previously well-defined 
concepts such as parent and child.

\begin{definition}
\defname{$x$-child}
Define an $x$-child of a node to be either of the two nodes which represent
the regions created when the original node's representative region is divided
in half along the $x$-axis.
\end{definition}

Note, there is both an left $x$-child, and right $x$-child.

\begin{definition}
\defname{$x$-parent}
Define the $x$-parent of a node to be the node for whom the original node
is an $x$-child.
\end{definition}

\begin{definition}
\defname{$x$-ancestor}
Define an $x$-ancestor to be any node which can be reached by following
a series of $x$-parent relationships.
\end{definition}

\begin{definition}
\defname{$x$-descendant}
Define an $x$-descendant to be any node which can be reached by following
a series of $x$-child relationships.
\end{definition}

Analogous definitions exist for $y$-child, $y$-parent, $y$-ancestor, and
$y$-descendant.
In addition to $x$-ancestors and $x$-descendants, we define additional
nodes to be simply ancestors and descendants.

\begin{definition}
\defname{Ancestor}
Define an ancestor to be any node which can be reached by following
a series of $x$-parent and/or $y$-parent relationships.
\end{definition}

\begin{definition}
\defname{Descendant}
Define an descendant to be any node which can be reached by following
a series of $x$-child and/or $y$-child relationships.
\end{definition}

Note that all nodes can have an $x$-parent, a $y$-parent, two $x$-children,
and two $y$-children, except for extremal nodes.
Therefore, the node structure must include additional pointers to
accommodate for this.  As with the original segment tree, each node 
keeps a list of the segments whose canonical representation includes
the node.

Some additional properties can be seen from these definitions

\begin{itemize}
\item{The $x$-parent of the $y$-parent of a node is the same node as
the $y$-parent of its $x$-parent.}
\item{The $x$-child of the $y$-child of a node may be the same node as
the $y$-child of its $x$-child, but only if the corresponding choices 
for left and right children are made.}
\end{itemize}

\section{Visualization}

We find it useful to visualize the unified segment tree as a diamond, where
the root of the data structure is at the top of the diamond.  
We divide the diamond into units, such that all nodes representing
a rectangle of the same shape are located in the same unit.

The two $x$-children of any node appear together in the same unit below 
and to the left of their $x$-parent. The two $y$-children of any
node appear together in the same unit below and to the right of their
$y$-parent.  Thus, each horizontal row of the diamond 
has double the number of nodes per unit, as the row above it.
See Figure \ref{fig:segtreevisualize}.

\begin{figure}
\centering
\includegraphics[height=6.2cm]{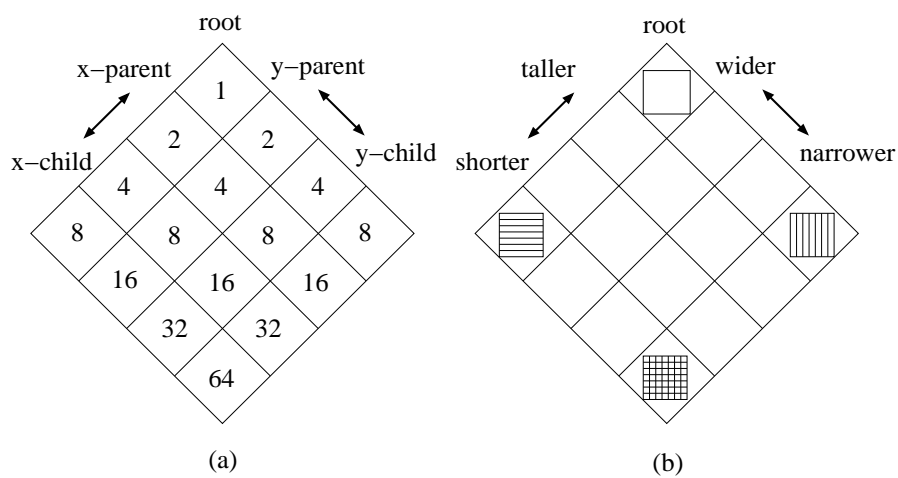}
\caption{
(a) A diamond-shaped visualization of parent child relationships, and the 
number of nodes in each unit of the diamond.  (b) The rectangles which 
are represented by the nodes in each unit of the diamond.
}
\label{fig:segtreevisualize}
\end{figure}

It is possible to see the $xy$-segment tree and the $yx$-segment tree
embedded in the diamond representation of the unified segment tree.
See Figure \ref{fig:embeddedtrees1}.

\begin{figure}
\centering
\includegraphics[height=4.2cm]{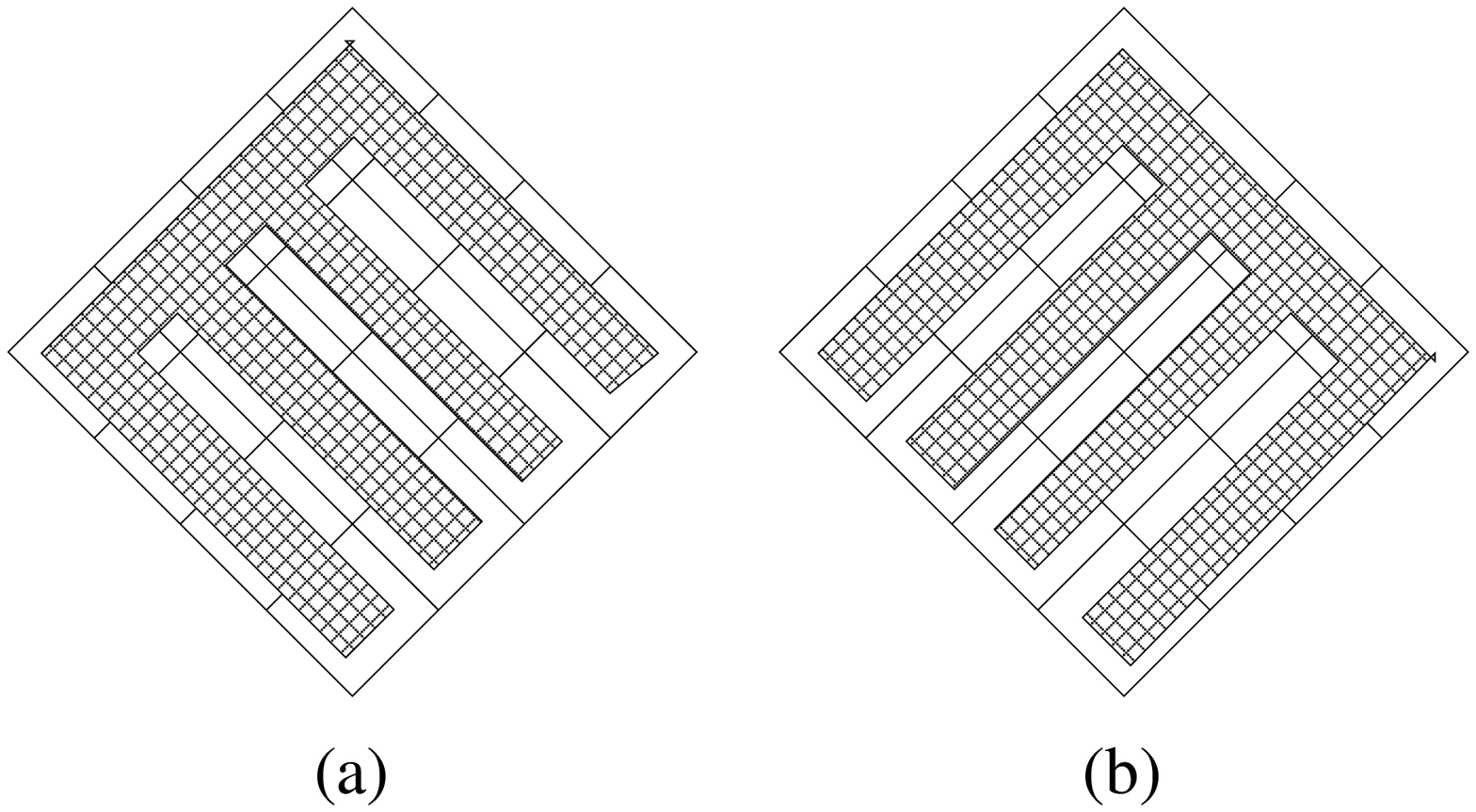}
\caption{
(a) The $xy$-segment tree embedded into the unified segment tree.
(b) The $yx$-segment tree embedded into the unified segment tree.
}
\label{fig:embeddedtrees1}
\end{figure}

Using this visualization, the ancestors of a node appear in the 
diamond shaped region above the node.  The descendants appear in the 
diamond shaped region below the node.  See Figure \ref{fig:ancestors}.
The number of ancestors can be bounded as follows.

\begin{theorem}
A node has at most one ancestor per unit of the diamond.
\end{theorem}
\begin{proof}
Assume there are two ancestors of a node within the same unit.
These two nodes must represent rectangles of the same shape, 
because they are within the same unit.
The representative rectangles must contain a common point,
since the nodes have a common descendant.
The representative rectangles must not partially overlap,
by Property \ref{prop:nooverlap} of Segment Trees (see Appendix).
Therefore, the rectangles, and the nodes representing them, must be the same.
\end{proof}

\begin{figure}
\centering
\includegraphics[height=6.2cm]{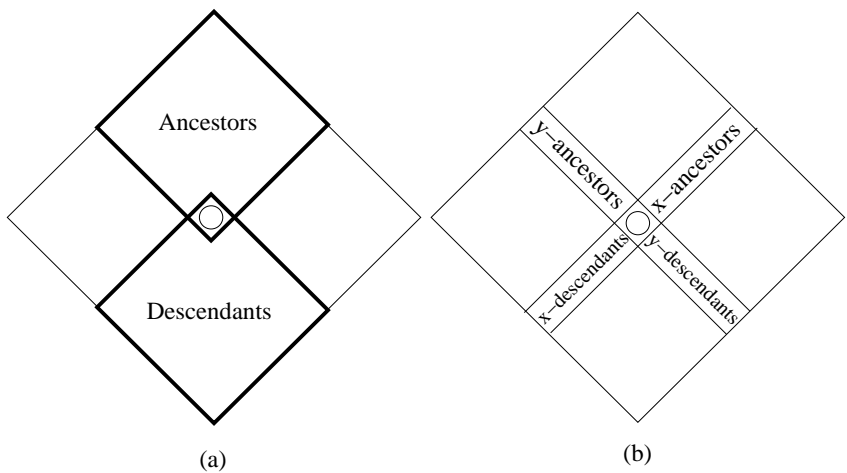}
\caption{
(a) The location of the ancestors and descendants of a node within the 
diamond.  (b) The location of x-ancestors, y-ancestors, x-descendants, 
and y-descendants within the diamond.
}
\label{fig:ancestors}
\end{figure}

We think it is interesting that the nodes taken from the central
vertical column of the diamond form a quad tree.
See Figure \ref{fig:embeddedtrees2} for a depiction of an embedding
this, as well of a k-D tree.

\begin{figure}
\centering
\includegraphics[height=5.1cm]{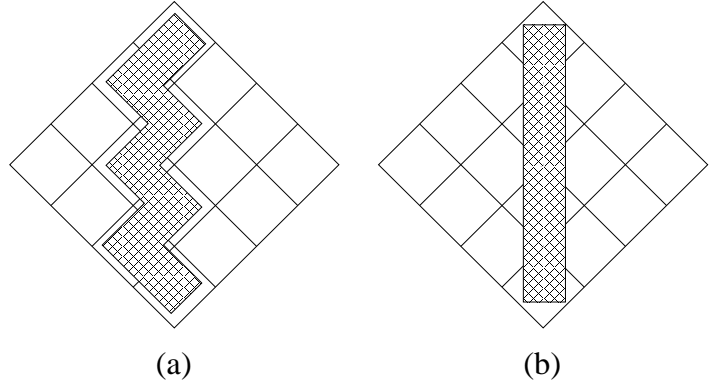}
\caption{
(a) A k-D tree embedded in the unified segment tree.
(b) A quad tree embedded in the unified segment tree.
}
\label{fig:embeddedtrees2}
\end{figure}

\section{Relationships between Nodes}

Here we look at some interesting relationships between the
nodes of a unified segment tree

\begin{theorem}
A node is an descendant of another node, if and only if it
represents a rectangle that falls entirely within that represented
by the ancestor node.
\label{thm:rectinside}
\end{theorem}
\begin{proof}
First, consider two nodes in the segment tree, one of which is a
descendant of the other.  The rectangle represented by the descendant was 
formed by successively subdividing the rectangle represented by the ancestor.
Therefore, the rectangle represented by a descendant falls entirely
within the rectangle represented by any of its ancestors.

Next, consider two rectangles represented by two different nodes
of the tree, such that one rectangle falls entirely within the other.
Follow the $x$-parents of the node representing the smaller rectangle,
until a node is found which has the same size in the $x$ direction
as the larger rectangle.  This must have the same $x$-coordinates as the
larger rectangle, otherwise Property \ref{prop:nooverlap} would be violated.
From there, follow $y$-parents until a node is found which has the same
size in the $y$ direction.  This node must have the same $y$-coordinates
as the larger rectangle, for the same reason.  Therefore it must be
the node which represents the larger rectangle, and the node representing
the smaller rectangle must be a descendant of the node representing the
larger rectangle.
\end{proof}

\begin{theorem}
An $x$-ancestor of a node is a $y$-ancestor of another node, if and only
if the representative rectangles of the two descendants completely cross 
over each other, one in the $x$ direction, and the other in the $y$ direction.
\label{thm:rectcross}
\end{theorem}
\begin{proof}
Consider two nodes, such that an $x$-ancestor of one is a $y$-ancestor of
the other.  The rectangle represented by the common ancestor of the two 
nodes can be formed by expanding one of the original two rectangles in the
$x$ direction, or by expanding the other in the $y$ direction.  Therefore,
the rectangle of the ancestor must be completely spanned in $y$ direction
by the first rectangle, and completely spanned in the $x$ direction by the
other.  Therefore the two rectangles must completely cross each other.

Next, consider two nodes which represent rectangles that completely
cross over each other.  The smallest rectangle enclosing both original
rectangles can be found by expanding one rectangle in the $x$-direction
or by expanding the other rectangle in the $y$-direction.
Therefore, the enclosing rectangle is represented by an $x$-ancestor 
of one of the original nodes, and a $y$-ancestor of the other node.
\end{proof}

\begin{theorem}
Two rectangles which are represented by nodes in a unified segment
tree may only intersect in one of two ways.  Either one rectangle
can be completely inside of the other, or the two rectangles can
completely cross over each other, one in the $x$ direction, and
the other in the $y$ direction.
\label{thm:rectint}
\end{theorem}
\begin{proof}
By enumerating the possible rectangle intersections, we
can see that all other intersections would violate 
Property \ref{prop:nooverlap}.  See Figures \ref{fig:rectint}
and \ref{fig:impossibleint} for a visual depiction of 
the possible and impossible intersections.
\end{proof}

\begin{figure}
\centering
\includegraphics[height=2.8cm]{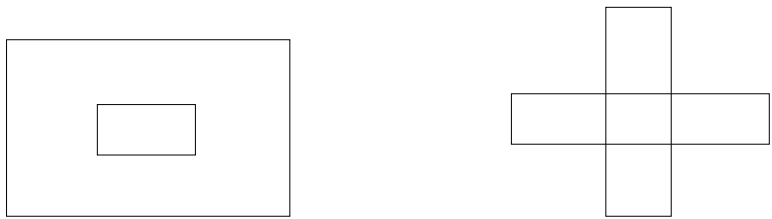}
\caption{
Possible intersections between rectangles represented
by nodes in a unified segment tree.
}
\label{fig:rectint}
\end{figure}

\begin{figure}
\centering
\includegraphics[height=2.8cm]{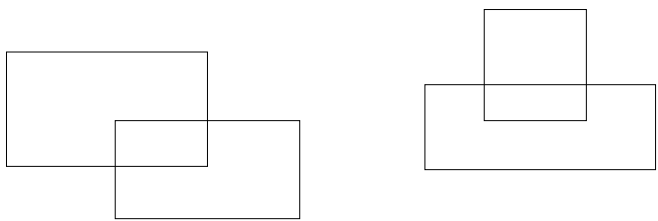}
\caption{
Impossible intersections between rectangles represented
by nodes in a unified segment tree.
}
\label{fig:impossibleint}
\end{figure}

\section{Analysis}

Here we analyze the performance of the unified segment tree.
We show bounds on the size of the data structure after
$n$ rectangles have been inserted.
Also we show bounds on the running time of the standard
segment tree operations.

\begin{theorem}
After the insertion of $n$ rectangles into a unified
segment tree, the deepest $x$-descendant of the root 
and the deepest $y$-descendant of the root have a maximum 
depth of $O(\log n)$.
\end{theorem}
\begin{proof}
Recall that the unified segment tree was created from
an $xy$-segment tree, and a $yx$-segment tree.
The outer segment trees of these two trees exactly comprise 
the $x$-descendants and the $y$-descendants of the root.
By Property \ref{prop:height}, these two trees can
have a maximum height of $O(\log n)$.
\end{proof}

\begin{corollary}
Any node in a unified segment tree can have a maximum 
of $O(\log n)$ $x$-ancestors and a maximum of 
$O(\log n)$ $y$-ancestors.
\end{corollary}

\begin{theorem}
Any node in a two dimensional unified segment tree can have 
a maximum of $O(log^2 n)$ ancestors.
\end{theorem}
\begin{proof}
A node can have a maximum of $O(\log n)$ $x$-ancestors, and each 
of these nodes can have a maximum of $O(\log n)$ $y$-ancestors.
All ancestors can be reached along one of these routes.
\end{proof}

\begin{theorem}
The canonical representation of a rectangle in a two dimensional
segment tree is comprised of a maximum of $O(\log^2 n)$
subrectangles.
\end{theorem}
\begin{proof}
Any rectangle is decomposed into a maximum of $O(\log n)$
regions in the $x$ direction by Property \ref{prop:maxcanonical}
of Segment Trees (see Appendix).
Each of these regions is further subdivided into
a maximum of $O(\log n)$ subregions in the $y$ direction.
\end{proof}

\begin{theorem}
All nodes representing the canonical subregions of a rectangle 
have a maximum of $O(\log^2 n)$ ancestors in a two dimensional
unified segment tree.
\label{thm:maxancestors}
\end{theorem}
\begin{proof}
The limit on the number of the ancestors of the 
canonical representation exists because there can be no more 
than 16 ancestors of a given shape.  Consider that there are 17 
ancestors of a particular shape in the tree.  
By Property \ref{prop:nooverlap}, these shapes cannot overlap
in their $x$-coordinates, or their $y$-coordinates, unless the
coordinates are the same.  Since there are 17 distinct sets
of $x$ and $y$-coordinates, there must be at least 5 distinct 
pairs of $x$-coordinates, or 5 distinct pairs of $y$-coordinates.
Assume, without loss of generality that there are 5 distinct pairs of
$x$-coordinates.  Consider the node representing the middle of 
these 5 pairs of coordinates.  The $x$-parent of the node representing 
this rectangle must be located entirely within the original rectangle.  
Therefore there is no reason to include the middle rectangle, or any of 
its descendants in the canonical representation.
\end{proof}

\begin{corollary}
Insertion of a rectangle into a two dimensional segment tree
requires $O(\log ^2 n)$ time.
\end{corollary}

\begin{corollary}
A two dimensional unified segment tree requires $O(n\log^2 n)$ space
to store $n$ rectangles.
\end{corollary}

\begin{theorem}
A point query of a two dimensional unified segment tree returns
the list of enclosing rectangles in $O(\log^2 n + k)$ time
where $k$ is the number of rectangles reported.
\end{theorem}
\begin{proof}
A point is represented in a unified segment tree at the deepest node.
All enclosing rectangles are represented by $O(\log^2 n)$ ancestors of this 
node.  Thus it is sufficient to report all rectangles stored in all ancestors.
\end{proof}

\section{Higher Dimensions}

Most every aspect of the two dimensional unified segment tree can
be generalized into higher dimensions in a straightforward way.
Here we mention a few of the differences.

\begin{itemize}
\item{Nodes may have $d$ parents and $2d$ children.}
\item{A node may have up to $O(\log^d n)$ ancestors.}
\item{Insertion requires $O(\log^d n)$ time.}
\item{Query requires $O(\log^d n + k)$ time.}
\item{The data structure occupies $O(n\log^d n)$ space.}
\end{itemize}

We note that the space and the running time of the operations could be 
sped up by a logarithmic factor, through the use of fractional cascading.
However, our objective here is to demonstrate that this data structure is 
a more natural representation of the data, rather than to achieve the
maximum speedup.  Therefore, we do not investigate the use of fractional 
cascading here.

\section{The Rectangle Intersection Problem}
The rectangle intersection problem is a classic problem dating
back to the early days of computational geometry. 
The problem has been studied by many authors, and numerous
variations have been inspired \cite{KapeliosPPST95,KaplanMT03}.

Although multiple definitions have been given for the name
``rectangle intersection problem'', we use the following definition.  
Store a set of $n$ axis-parallel rectangles in such a
way, that a rectangle query will efficiently report a list
of all rectangles from the set which intersect the query
rectangle.

Edelsbrunner and Maurer gave one of the first general purpose
algorithms solving this problem in 1981 \cite{EdelsbrunnerM81}.
Their method involves storing the rectangles in a combination of
segment trees and range trees, and these are queried to detect 
four kinds of intersections between rectangles.
In higher dimensions, they sweep across the problem, solving a 
lower dimensional problems during the sweep, and thereby giving 
an overall solution that is recursive by dimension.

Edelsbrunner demonstrated two further methods to solve the problem
in 1983 \cite{Edelsbrunner83I}.  The first of these involves
Storing the $d$-dimensional rectangle in a $2d$-dimensional range
tree, and performing an appropriate range query.  For the
second, a data structure called the ``rectangle tree'' is developed, 
and this is queried similarly to the range tree.

Here we show yet another way to solve this problem using an
augmented version of our unified segment tree.  
Our method does not claim a speedup over the previous methods.  
However, we feel the unified segment tree uses a more natural 
representation of the data than the range tree or rectangle tree.
Additionally, given the machinery that we have already developed
in this paper, the algorithm is quite straightforward.

\subsection{Augmenting the Unified Segment Tree}
For the purpose of solving the rectangle intersection problem, we 
augment the unified segment tree with the following additional
data:
\begin{itemize}
\item{A list of rectangles stored in all descendants of the node.}
\item{A list of rectangles stored in all $x$-descendants of the node.}
\item{A list of rectangles stored in all $y$-descendants of the node.}
\end{itemize}

Thus a node of a two dimensional augmented segment tree contains 
the following information.
\begin{verbatim}
struct NODE {
    struct NODE * xparent, yparent;
    struct NODE * leftxchild, rightxchild;
    struct NODE * leftychild, rightychild;
    Segment storedHere[];
    Segment storedInDescendants[];
    Segment storedInXDescendants[];
    Segment storedInYDescendants[];
}
\end{verbatim}

When a segment is inserted, its identifier must be inserted
into all canonical nodes, and all ancestors of the canonical nodes, 
in the appropriate lists.  
In two dimensions, this does not affect the asymptotic running 
time of the insert operation.  However, in $d$ dimensions, there can
exist $2^d$ separate lists, so an alternate method of storage may be
desirable if $d$ is large.

\subsection{Rectangle Query Algorithm}
A rectangle query operation returns a list of all rectangles which 
intersect the given query rectangle.  
Our algorithm for rectangle query first divides the query rectangle 
into its canonical regions.  It then performs a query on each rectangle 
individually, reporting the union of the rectangles found.

  Note that the same rectangle might be found in multiple places, so care 
must be taken to avoid reporting duplicates, if that is undesirable.
If duplicates are reported it may also adversely affect the running time.

Recall from Theorem \ref{thm:rectint} that rectangles can only 
intersect if one is completely inside the other, or if they
completely cross over each other, one in the $x$ direction, and the
other in the $y$ direction.  Therefore, it is sufficient to report
the rectangles described in Theorems \ref{thm:rectinside} and 
\ref{thm:rectcross}.  

This gives us the following very straightforward algorithm:
\begin{enumerate}
\item{Report all rectangles stored in ancestors of the node.}
\item{Report all rectangles stored in the descendant list of the node.}
\item{Report all rectangles stored in the $x$-descendant list of 
a $y$-ancestor of the node.}
\item{Report all rectangles stored in the $y$-descendant list of 
an $x$-ancestor of the node.}
\end{enumerate}
This information is available in the ancestors of the canonical
nodes of the query rectangles.  So only $O(\log^2 n)$ nodes
need to be accessed.  Again care must be taken to avoid
reporting duplicates.

\subsubsection*{Acknowledgments.} 
The author is grateful to Stefan Langerman and John Iacono
for their helpful discussions.

\bibliographystyle{abbrv}

\bibliography{bibliography}

\section*{Appendix: The Segment Tree}

We include a short description of the segment tree data structure, 
originally described by Bentley in 1977 \cite{Bentley77}.

A segment tree is a a balanced binary tree which stores
a set of line segments.  The data structure supports point
queries, which return the list of all segments containing
the query point.

   Each node of the tree represents a contiguous section of the line, 
may be bounded or unbounded, and may be open or closed.
The root of the tree represents the entire available line.
The two children of any node represent two mutually exclusive
subsections of the line, and the union of these two
subsections is equivalent to the section represented by the parent.

The points where the line may be divided are normally drawn from
the endpoints of the segments stored.  Thus, the entire set of segments
is assumed to be known before the tree is built, although methods have 
been developed to make the segment tree dynamic 
\cite{vanKreveldO89,vanKreveldO93}.

Before a segment is inserted, it must first be divided into the unique 
set of mutually exclusive subsegments, whose union is equivalent to the 
original segment, and which have representative nodes in the tree.
This set is called the ``canonical representation'' of the segment.
An identifier for the original segment is then stored at every node
which represents one of the subsegments.

Segment trees have many well known properties. 
Here we show a proof of the some properties which we 
will use later in the paper.
These apply to a segment tree which is designed to hold $n$ segments.

\begin{property}
The height of a segment tree is at most $O(\log n)$.
\label{prop:height}
\end{property}
\begin{proof}
The segment tree is well-balanced, and therefore its height can be optimal.
\end{proof}

\begin{property}
The canonical representation of a segment has, at most, $O(\log n)$ segments.
\label{prop:maxcanonical}
\end{property}
\begin{proof}
The limit on the number of segments in the canonical representation 
exists because there can be no more than 2 segments at a given depth
of the segment tree.
Consider that there are 3 segments at a particular depth
in the tree.  The parent of the middle segment must be located
entirely within the original segment.  Therefore there is no reason
to include the middle segment in the canonical representation.
\end{proof}

\begin{property}
The ancestors of the canonical representation of a segment
are comprised of, at most, $O(\log n)$ nodes.
\label{prop:maxancestors}
\end{property}
\begin{proof}
The limit on the number of the ancestors of the 
canonical representation exists because there can be no more 
than 4 ancestors at a given depth.  Consider that there are 5 
ancestors at a particular depth in the tree.  
The parent of the middle segment must be located
entirely within the original segment.  Therefore there is no reason
to include the middle segment, or any of its descendants in the 
canonical representation.
\end{proof}

\begin{property}
$O(\log n)$ time is required for insertion and query in a segment tree.
\end{property}
\begin{proof}
The only nodes processed during an insertion 
or query are the ancestors of the canonical set.
So the running time follows directly from Property \ref{prop:maxancestors}.
\end{proof}

\begin{property}
The line segments represented by any two nodes in a segment tree are
either mutually exclusive, or one is contained in the other.
\label{prop:nooverlap}
\end{property}

\begin{proof}
Consider any two nodes in a segment tree.  Either one of these nodes
is an descendant of the other, or neither is a descendant of the other.

If one node is a descendant of the other, then the line segment represented 
by the descendant node is defined by repeatedly subdividing the line segment 
represented by the ancestor.  Thus the descendant line segment is contained 
in the ancestor line segment.

If there is no descendant relationship between the two nodes, then they must
have a deepest common ancestor, which is neither of the two nodes.  
One of the two nodes must be in the left subtree of the deepest common
ancestor, and one must be in the right subtree.  
The line segments represented by the left and right children of this node 
are be mutually exclusive, so the line segments represented
by one node from each subtree must be mutually exclusive.
\end{proof}

In particular, we note that if two nodes, $N_a$ and $N_b$, in the segment tree
represent segments $[min_a, max_a]$ and $[min_b, max_b]$, then it
cannot be the case that $min_a < min_b < max_a < max_b$.

\end{document}